\documentclass[pra, twocolumn, superscriptaddress]{revtex4-1}
\usepackage{amsmath}
\usepackage{graphicx, color}
\usepackage{amsthm}
\usepackage{amssymb}
\usepackage{relsize}
\usepackage{braket}
\usepackage{mathtools}
\usepackage[caption=false]{subfig}
\usepackage[ colorlinks = true,              linkcolor = blue, urlcolor  =
blue,              citecolor = red,              anchorcolor = green,
]{hyperref}%

\newcommand{\bl}[1]{\textcolor{blue}{#1}}

\newcommand\numeq[1]%
  {\stackrel{\scriptscriptstyle(\mkern-1.5mu#1\mkern-1.5mu)}{=}}

\newtheorem{theorem}{Theorem}
\newtheorem{definition}{Definition}
\newtheorem{proposition}[theorem]{Proposition}

\begin{document}

\title{Quantum State Discrimination Circuits Inspired by\\Deutschian Closed Timelike Curves}
\author{Christopher Vairogs}
\affiliation{Department of Mathematics and Department of Physics, University of Florida, Gainesville, Florida 32611, USA}
\affiliation{Hearne Institute for Theoretical Physics, Department of Physics and Astronomy, and Center for Computation and Technology, Louisiana State University, Baton Rouge, Louisiana 70803, USA}
\author{Vishal Katariya}
\affiliation{Hearne Institute for Theoretical Physics, Department of Physics and Astronomy, and Center for Computation and Technology, Louisiana State University, Baton Rouge, Louisiana 70803, USA}
\author{Mark M. Wilde}
\affiliation{Hearne Institute for Theoretical Physics, Department of Physics and Astronomy, and Center for Computation and Technology, Louisiana State University, Baton Rouge, Louisiana 70803, USA}
\date{\today}

\begin{abstract}
It is known that a party with access to a Deutschian closed timelike curve (D-CTC) can perfectly distinguish multiple non-orthogonal quantum states. In this paper, we propose a practical method for discriminating multiple non-orthogonal states, by using a previously known quantum circuit designed to simulate D-CTCs. This method relies on multiple copies of an input state, multiple iterations of the circuit, and a fixed set of unitary operations. We first characterize the performance of this circuit and study its asymptotic behavior. We also show how it can be equivalently recast as a local, adaptive circuit that may be implemented simply in an experiment. Finally, we prove that our state discrimination strategy achieves the multiple Chernoff bound when discriminating an arbitrary set of pure qubit states. 
\end{abstract}

\maketitle

\tableofcontents

\section{Introduction} \label{sec:introduction}

Closed timelike curves (CTCs) arise as solutions to the Einstein field equations in general relativity. While the existence of CTCs is unverified, they bring up the possibility of time travel and the paradoxes associated with it~\cite{Godel}. To better understand the properties of these objects, several quantum information theoretic models of CTCs have been proposed \cite{Deutsch, Svetlichny, P-CTC-Lloyd-et-al, T-CTC-Allen}.  

One such CTC model is that given by Deutsch \cite{Deutsch}, where paradoxes associated with time travel using CTCs are resolved by a self-consistency condition.
This self-consistency condition introduces a non-linearity in the evolution of a quantum state through a Deustchian CTC (D-CTC). Standard quantum mechanics demands that the evolution of an arbitrary state is linear, which places restrictions on physical evolutions, such as the no-cloning theorem \cite{Park1970,nat1982,D82} and the Heisenberg uncertainty principle.

Thus, in contrast to standard quantum mechanics, this non-linearity allows for many remarkable characteristics associated with D-CTCs beyond what is allowed by standard quantum mechanics.
D-CTCs may be utilized to violate the no-cloning theorem~\cite{No-Cloning-Violation-BWW}, the Holevo bound~\cite{Heisenberg-Holevo-Violation-BHW}, the second law of thermodynamics~\cite{CTC-SecondLaw}, and enable quantum computers to solve problems in the computational complexity class PSPACE~\cite{PSPACE-Watrous}. (However, note that these claims have been debated in the literature \cite{BLSS09,CM10}).

The aspect of D-CTCs that we are most interested in here is their use in perfectly distinguishing multiple non-orthogonal quantum states, violating Heisenberg's uncertainty principle~\cite{Heisenberg-Holevo-Violation-BHW}. We use ideas contained in the D-CTC-assisted state discrimination method to create a practical, iterative state discrimination circuit that works by approximating the behavior of a D-CTC. 

Even though D-CTCs are inaccessible, we may simulate the evolution of the state of a system traveling along a D-CTC. Such simulations are important to us not only because they allow us to gain a better understanding of the properties of D-CTCs in an accessible setting, but also because they enable us to exploit their unique characteristics for applications. 
Simulating a D-CTC is directly related to computing the fixed point of a quantum channel, which is a difficult task \cite{PSPACE-Watrous}. One D-CTC simulation method uses polarization-encoded photons as qubits~\cite{Simulations-RBMWR}, which involves computing the self-consistent solution for the state of a system traveling along a D-CTC. This computation is practical for simple quantum systems, but it becomes prohibitively expensive for larger systems. Circumventing this issue, the authors of \cite{BW-Simulations-of-CTCs} proposed a method for simulating CTCs that uses an iterative quantum circuit, with the circuit approaching the behavior of a D-CTC with an increasing number of iterations. It is also ``self-contained'' in the sense that it does not involve the discarding of experimental data, unlike that in~\cite{Simulations-RBMWR}.

Strategies for discriminating non-orthogonal quantum states have been analyzed in
many different contexts. In this paper, we  restrict our attention to the minimum-error discrimination of pure states. The case of discriminating two quantum states has
been well studied \cite{Helstrom, Two-State-Disc-Acin-et-al, ThinkingGlobalHiggins, Noisy-Qubits-Flatt-et-al}. Optimal minimum-error approaches for discriminating multiple quantum states
have been characterized in a variety of specific cases \cite{GeomUniformBan, MirrorSymAndersson, Bae-Multiple-Qubit-States, Ha2013, Slussarenko2017, Weir2017}. A strategy using a
theoretical apparatus for distinguishing multiple arbitrary quantum states has been
proposed in \cite{Blume-Kohout-QDataGathering}. Our method is similar in that it applies to the general task of discriminating multiple (and possibly non-orthogonal)
quantum states.

The major contribution of our paper is a practical state discrimination strategy for multiple non-orthogonal pure states that combines the D-CTC simulation circuit of \cite{BW-Simulations-of-CTCs} and the CTC-assisted state discrimination strategy of \cite{Heisenberg-Holevo-Violation-BHW}. We briefly state our strategy here. Assume that the set of states to be discriminated is $\{|\psi_i\rangle\}_{i=0}^{N-1}$ and that we are given $n$ copies of an unknown state randomly selected from this set. Assume that measurements are made in the basis $\{|i\rangle\}_{i=0}^{N-1}$. Suppose that we have access to a set of unitary operations $\{U_i\}_{i=0}^{N-1}$ such that $U_i|\psi_i\rangle = |i\rangle$. Suppose that the $n$ copies of the unknown state are ordered. Perform a measurement on the first copy of the unknown state. If outcome $j$ is obtained, then apply the unitary operation $U_j$ to the second copy of the unknown state and perform a measurement on the resulting state. If outcome $k$ is obtained, then apply the unitary operation $U_k$ to the third copy of the unknown state. Repeat this procedure until a unitary operation has been applied to all $n$ states. If outcome $l$ is obtained after measuring the $n$-th state, then claim that the unknown state is the state $|\psi_l\rangle$.

Our strategy proves advantageous because it relies only on a fixed set of local operations and may be implemented by performing measurements on each copy of a state, with each successive measurement being determined by the outcome of the previous one, i.e., a $\textit{local adaptive strategy}$ \cite{Two-State-Disc-Acin-et-al, ThinkingGlobalHiggins, Noisy-Qubits-Flatt-et-al}. In this way, our state discrimination strategy is amenable to a practical experimental implementation. [See Figure~\ref{fig:adaptive-circuit-3} for a schematic of the method.]

Furthermore, we calculate the asymptotic rate of decay of the average probability of error of our state discrimination circuit, which we use to show that our state discrimination scheme attains the fundamental limit, i.e., the multiple Chernoff bound \cite{NS11,Li-Chernoff-Bound}, for the general task of discriminating an arbitrary set of pure qubit states. This is another desirable property that our state discrimination scheme possesses.

The rest of our paper is structured as follows. We provide some preliminaries and set up notation in Section~\ref{sec:prelims}.  Then, in Section~\ref{subsec:prelim-CTC-description}, we provide our state discrimination circuit and also show how it can be implemented as a local, adaptive circuit. In Section~\ref{sec:error-analysis}, we calculate its average probability of error in distinguishing states. We show that this probability of error converges to zero in the limit of infinitely many iterations of the circuit. Finally in Section~\ref{sec:examples}, we consider two examples, and show how our scheme achieves the multiple Chernoff bound when discriminating an arbitrary set of pure qubit states.

\section{Preliminaries} \label{sec:prelims}

\subsection{State Discrimination}

We first describe the problem of state discrimination considered in this paper. The goal is to distinguish the states in the set $\{ \rho_i \}_{i=0}^{N-1}$, where the state $\rho_i$ is chosen with probability $p_i$. The minimum error approach to state discrimination may be pictured as the following game between Alice and Bob. Alice and Bob agree on the set $\{ \rho_i \}_{i=0}^{N-1}$ of quantum states and probability distribution $\{p_i\}_i$  that they will use. Alice prepares a state $\rho_j$ from that set with probability $p_j$ and sends it to Bob. Bob then, in an attempt to identify Alice's state, performs a measurement described by the set $\{M_i\}_i$ of measurement operators. He guesses that the state Alice prepared is $\rho_k$ if he measures outcome $k$. Bob's goal is to find the measurement that minimizes his average probability of error, defined as follows:
\begin{equation}
p_e \coloneqq  \sum_{k=0}^{N-1} p_k \sum_{j\neq k} 
\operatorname{Tr}\big\{M_j\rho_k\big\} =
 1 - \sum_{k=0}^{N-1} p_k 
\operatorname{Tr}\big\{M_k\rho_k\big\}.
\end{equation}
\noindent In the case when $N=2$, the Helstrom measurement is an optimal measurement \cite{Helstrom}.

An alternative approach to the state discrimination problem is to assume that Bob has $n$ available copies of the state $\rho_j$ that Alice selects. In this case, an optimal measurement is a collective measurement on all $n$ states $\rho_j^{\otimes n}$. In the limit of large $n$, the optimal error probability $p_e^{\operatorname{opt}}$ decays exponentially as
\begin{equation}
p_e^{\operatorname{opt}} \sim e^{-n\xi^{\operatorname{opt}}},
\end{equation} 
where the value $\xi^{\operatorname{opt}}$ is known as the multiple Chernoff bound \cite{Li-Chernoff-Bound} and is given by
\begin{equation}\label{eq:chernoff-bound-def}
\xi^{\operatorname{opt}} = 
 -\ln\! \left[\max_{i\neq j}\min_{0\leq s\leq 1}\operatorname{Tr}\{\rho_i^s\rho_j^{1-s}\}\right].
\end{equation}
The multiple Chernoff bound places a fundamental limit on how fast the error probability decays for a multiple-copy state discrimination scheme \cite{Li-Chernoff-Bound}. See \cite{NS11} for the special case when all of the states in the set $\{\rho_i\}_i$ are pure.


\subsection{Deutschian Closed Timelike Curves}

Next, we describe the model for CTCs that is applicable to our work---namely, the Deutschian (D-CTC) model~\cite{Deutsch}. We note that there also exist other models for quantitatively describing the behavior of CTCs, namely post-selected quantum teleportation CTCs (P-CTCs)~\cite{P-CTC-Lloyd-et-al} and transition probability CTCs (T-CTCs)~\cite{T-CTC-Allen}. 

The D-CTC model involves two sub-systems: the $\textit{chronology-respecting}$ (CR) system~$S$, which does not travel through the CTC, and the $\textit{chronology-violating}$ (CV) system $C$, which does travel through the CTC. The different CTC models differ in the manner in which they resolve or avoid causality paradoxes. In the D-CTC model, this is accomplished by requiring the state of the chronology-violating system to be a fixed point of the evolution that results from the interaction between the CR and CV systems. 

Let $\sigma_C$ be the state of the CV system, and let $\rho_S$ be the state of the CR system. In Deutsch's model, systems $S$ and $C$ are assumed to be in a tensor-product state $\rho_S \otimes\sigma_C$ before they interact unitarily via the CTC. They then interact according to an interaction unitary~$V_{SC}$ before the CV system enters the future mouth of its wormhole, so that the state of the composite system after the evolution is  $V_{SC}(\rho_{S}\otimes\sigma_C)V_{SC}^{\dagger}$. We refer to states of the CR system as the CV system emerges from the past mouth of its wormhole and as it enters the future mouth of its wormhole as ``initial'' CR states and ``final'' CR states, respectively. We define ``initial'' CV states and ``final'' CV states similarly.

 The evolution of the CV system is represented by the quantum channel
 \begin{equation}\label{eq:CTC-map-definition}
     \mathcal{N}_{V,\rho}: \sigma_C \mapsto 
     \operatorname{Tr}_{S} \big\{V_{SC}(\rho_{S}\otimes\sigma_C)V_{SC}^{\dagger}\big\},
 \end{equation}
 which maps each possible initial CV state to its corresponding final CV state when the initial CR state is~$\rho_S$. Furthermore, as stated earlier, the D-CTC model enforces a self consistency condition so as to avoid ``grandfather-like'' causality paradoxes. This requires that an arbitrary state $\sigma_C$ of the CV system is a fixed point of the quantum channel $\mathcal{N}_{V,\rho}$, i.e.,
\begin{equation}\label{eq:self-consistency-condition}
 \mathcal{N}_{V,\rho}(\sigma_C) = 
 \sigma_C.
\end{equation}

A solution $\sigma_C$ to~\eqref{eq:self-consistency-condition} always exists, although it is not necessarily unique~\cite{Deutsch}. The evolution of the CR system is represented by
\begin{equation} 
    \mathcal{M}_{V}: \rho_S \mapsto 
    \operatorname{Tr}_{C} \big\{V_{SC}(\rho_S\otimes\sigma_C)V_{SC}^{\dagger}\big\},
\end{equation}
which maps every possible initial CR state to its corresponding final CR state whenever the CV state is $\sigma_C$. 
Note that $\mathcal{M}_{V}(\rho_{S})$ depends not only on $\rho_{S}$, but also on $\sigma_C$, whose dependence on $\rho_{S}$ is given by~\eqref{eq:self-consistency-condition}. Therefore, $\mathcal{M}_{V}$ is not a linear map. The nonlinear nature of this evolution leads to the host of interesting properties mentioned earlier in Section~\ref{sec:introduction}.

\section{CTC-Inspired State Discrimination Circuit} \label{subsec:prelim-CTC-description}

Before we provide our CTC-inspired state discrimination circuit, we recall how to discriminate multiple non-orthogonal pure states using a D-CTC~\cite{Heisenberg-Holevo-Violation-BHW}. The set of states to be discriminated is $\{|\psi_i\rangle\!\langle\psi_i|\}_{i=0}^{N-1}$, and at the end of the D-CTC-assisted interaction, one may perform a measurement in the orthonormal basis $\{|i\rangle\}_{i=0}^{N-1}$. To do so, we require a set of unitaries $\{U_i\}_{i=0}^{N-1}$ such that 
$U_i|\psi_i\rangle = |i\rangle$ and $\langle j|U_i|\psi_j\rangle \neq 0$ for all $0\leq i, j\leq N-1, i \neq j$. Such a set of unitaries exists for every set $\{|\psi_i\rangle\!\langle\psi_i|\}_{i=0}^{N-1}$ of states \cite{Heisenberg-Holevo-Violation-BHW}. Let the interaction unitary $V_{SC}$ of the D-CTC be 
\begin{equation} \label{eq:brun-interaction-unitary}
V_{SC} = \left(\sum_{i=0}^{N-1} 
|i\rangle\!\langle i|_S \otimes ({U_i})_C \right)\circ\operatorname{SWAP},
\end{equation}
where SWAP is the unitary operator that swaps systems~$S$ and $C$.
The circuit representation of this unitary is shown in Figure~\ref{fig:ctc-figure}. The authors of \cite{Heisenberg-Holevo-Violation-BHW} demonstrated that $\sigma_C=|i\rangle\!\langle i |$ is the unique, self-consistent solution to \eqref{eq:self-consistency-condition} whenever $\rho_S =|\psi_i\rangle\!\langle\psi_i|$. That is, the D-CTC-assisted circuit can be used to map non-orthogonal states to distinct orthogonal basis states, and hence one can perfectly discriminate the non-orthogonal states in question.

To outline the functioning of the circuit, we briefly explain why each $\sigma_C = |i\rangle\!\langle i|$, for $0 \leq i \leq N -1 $, is a fixed point of $\mathcal{N}_{V,\rho}$, i.e., a solution to \eqref{eq:self-consistency-condition}. Suppose that $\rho_S= |\psi_i\rangle\!\langle \psi_i|$ and $\sigma_C = |i\rangle\!\langle i|$. The SWAP gate acts first and transforms the CR system to the state $|i\rangle\!\langle i|$. Next, the $U_i$ unitary is triggered. The unitary~$U_i$ acts on $|\psi_i\rangle\!\langle \psi_i|$, which leads to the self-consistency condition $\mathcal{N}_{V,\rho}(\sigma_C) =|i\rangle\!\langle i|$. Now, if one performs a measurement in the basis $\{|i\rangle\}_i$ on the final CR state, one may determine $\rho_S$ with certainty. Measurement outcome $j$ corresponds to the initial state $\rho_S = |\psi_j\rangle\!\langle \psi_j|$. Thus, using the D-CTC, in principle, we are able to discriminate perfectly the possibly non-orthogonal states $\{|\psi_i\rangle\!\langle\psi_i|\}_{i=0}^{N-1}$. 

\begin{figure} 
\includegraphics[width=\linewidth]{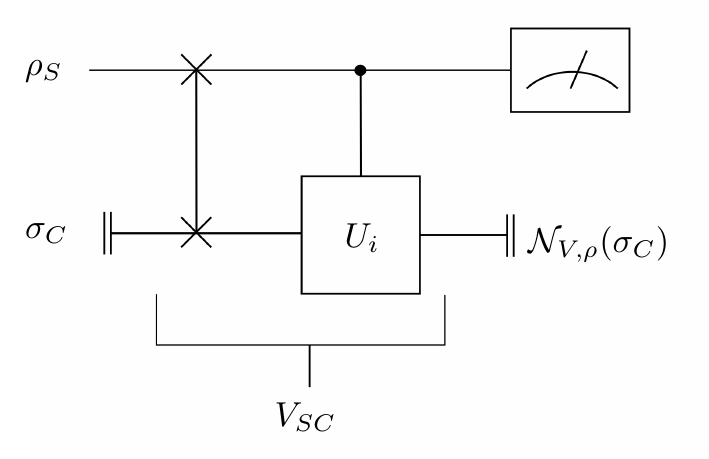}
\caption{The upper system is the CR system, and the lower system is the CV system. The past and future mouths of the wormhole are represented by the double bars on the left and right. The CR and CV systems interact according to the unitary $V_{SC}$, defined in \eqref{eq:brun-interaction-unitary}, before the CV system enters the future mouth of the wormhole.}
\label{fig:ctc-figure}
\end{figure}

The construction outlined above works if we have access to a D-CTC. In its absence, we can only construct iterative circuits that approximate its behavior. Our major contribution is one such circuit; i.e., we combine the CTC-assisted state discrimination scheme described above with a version of a quantum circuit given by \cite{BW-Simulations-of-CTCs} that simulates the behavior of a D-CTC. The circuit consists of multiple copies of three registers whose $n$th copies we label $G_n$, $S_n$, and $C_n$ (Figure~\ref{fig:dctc-circuit-simulation}). We will initialize each of the $G_n$ registers to the state vector $\sqrt{\gamma} \ket{0}  + \sqrt{1-\gamma} \ket{1} $. We find, however, that for our purposes, it suffices to set $\gamma = 0$. Each of the $S_n$ registers is initialized to the initial CR state $\rho_S$, and the $C_0$ register is initialized to a state $\omega$.
At every step of the circuit, the $G_n$, $S_n$ and $C_n$ systems interact via the controlled unitary
\begin{equation}
|0\rangle\!\langle 0|_{G_n}\otimes I_{S_n C_n}  + |1\rangle\!\langle 1|_{G_n}\otimes V_{S_n C_n}, \label{eq:controlled-unitary-defn}
\end{equation}
where $V_{S_n C_n}$ is the interaction unitary in \eqref{eq:brun-interaction-unitary} acting on the $S_n$ and $C_n$ registers. In other words, the procedure is the following: 
\begin{enumerate}
    \item Apply the controlled unitary between the $G_n$, $S_n$, and $C_n$ registers,
    \item discard the $S_n$ and $G_n$ registers, and
    \item load the resulting state of the $C_n$ system into the $C_{n+1}$ register, increment $n$ by one, and repeat.
\end{enumerate}

By applying this procedure, for every $\rho_S$, the state of the $C_n$ register converges to the fixed point of $\mathcal{N}_{V,\rho}$ as $n\rightarrow \infty$ \cite{BW-Simulations-of-CTCs}. Note that the rate of convergence is dependent on $V_{SC}$ and the state $\omega$ to which $C_0$ is initialized. We will discuss how to optimize the rate of convergence with respect to them. 

\begin{figure}
\includegraphics[width=\linewidth]{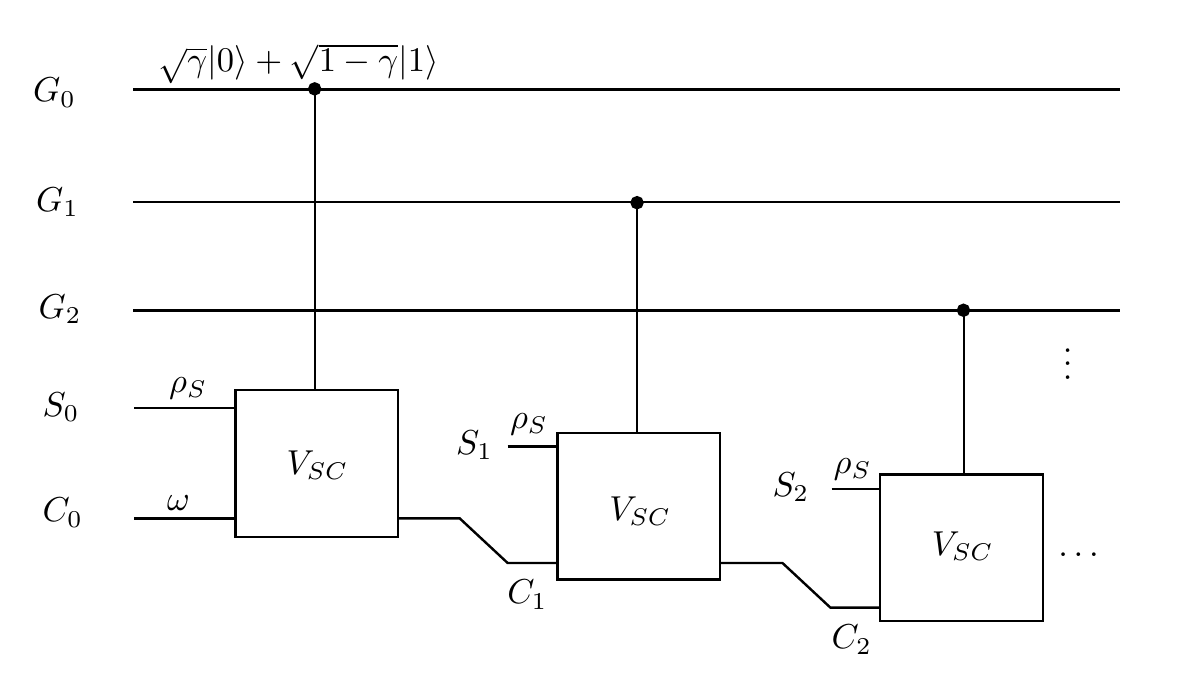}
\caption{This circuit simulates a D-CTC with interaction unitary $V_{SC}$. This figure shows three rounds of interaction. As the number of rounds $n$ increases, the simulation approaches the behavior of a D-CTC.}
\label{fig:dctc-circuit-simulation}
\end{figure}

To use this circuit, which is also shown in Figure~\ref{fig:dctc-circuit-simulation}, in order to discriminate the set $\{|\psi_i\rangle\!\langle\psi_i|\}_i$ of states, the interaction unitary $V_{SC}$ is set to the one in \eqref{eq:brun-interaction-unitary} and the input state~$\rho_S$ is restricted to belong to the set $\{|\psi_i\rangle\!\langle\psi_i|\}_i$. Since the state of the $C_n$ register converges to the fixed point of $\mathcal{N}_{V,\rho}$, the $C_n$ register  converges to the state $|j\rangle\!\langle j|$ if $\rho_S = |\psi_j\rangle\!\langle \psi_j|$ \cite{Heisenberg-Holevo-Violation-BHW}. Therefore, by implementing a standard basis measurement on the $C_n$ register after a suitably chosen $n$, we have a method for approximately discriminating the states in the set~$\{|\psi_i\rangle\!\langle\psi_i|\}_i$.


\begin{figure}[ht!]
\centering
\subfloat[\label{fig:adaptive-circuit-1}]{\includegraphics[width=\linewidth]{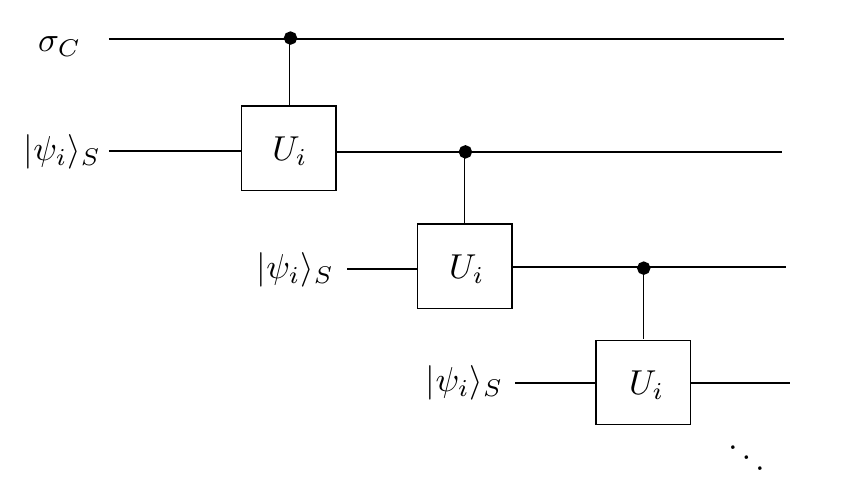}} \\

\subfloat[\label{fig:adaptive-circuit-2}]{\includegraphics[width=\linewidth]{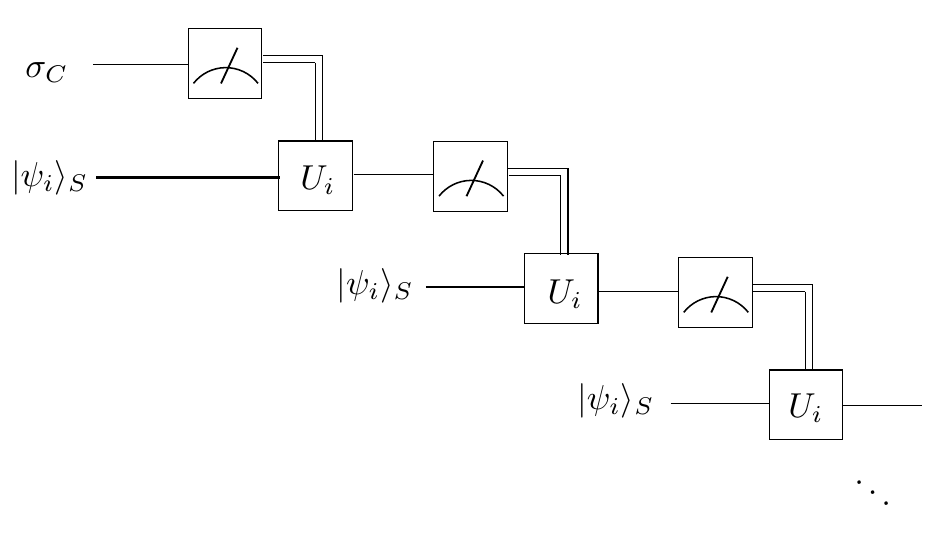}} \\

\subfloat[\label{fig:adaptive-circuit-3}]{\includegraphics[width=\linewidth]{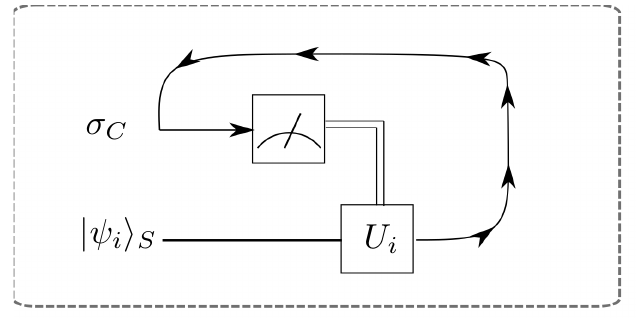}} \\

\caption{A series of simplifications of the original state discrimination circuit in Figure~\ref{fig:dctc-circuit-simulation} that ultimately leads to the simple iterative circuit in Figure~\ref{fig:adaptive-circuit-3}. First, we explicitly expand out $V_{SC}$ and set the state of the $G_n$ registers to $|1\rangle\! \langle 1|$, as this is optimal for performing our state discrimination scheme. This leads to the circuit in Figure~\ref{fig:adaptive-circuit-1}. We also note that the controlled-$U_i$ gate is short for the gate $\sum_i |i\rangle\! \langle i| \otimes U_i$. Further, as the controlling registers are effectively classical, the circuit can be simplified to Figure~\ref{fig:adaptive-circuit-2}. Finally, the iterative circuit is rewritten in the compact form of Figure~\ref{fig:adaptive-circuit-3}. The arrows do not depict a time-travel loop, but instead depict that the output of the unitary $U_i$ at one time interval is fed as input to the $C$ system at the next time interval.}
\label{fig:adaptive-figure-group}
\end{figure}

One of our major contributions is an equivalent recasting of our state discrimination circuit as a local, adaptive circuit. The adaptive circuit, which we describe below, consists of repeated iterations of a fixed two-register protocol and lends itself directly to experimental implementation. Another benefit of the adaptive circuit is that it is not necessary to update a quantum memory throughout the length of the circuit. The only information necessary to store in a quantum memory are the $n$ copies of the unknown state $|\psi_i\rangle$. 

We now explain how the original circuit in Figure~\ref{fig:dctc-circuit-simulation} can be recast in the simple, adaptive form of Figure~\ref{fig:adaptive-circuit-3}. To perform the state discrimination in an adaptive manner, we first assume that we have $n$ copies of the unknown state $\ket{\psi_i}$. By expanding out the $V_{SC}$ unitaries defined in~\eqref{eq:brun-interaction-unitary} and by applying the fact that the state~$\ket{1}$ suffices for each of the $G_n$ registers in Figure~\ref{fig:dctc-circuit-simulation}, it can be simplified to the circuit in Figure~\ref{fig:adaptive-circuit-1}. We then note that since the various $S$ registers are traced out at the end of the circuit, the coherent controls in Figure~\ref{fig:adaptive-circuit-1} can be equivalently replaced with classical controls. That further means that the classical control is effectively implemented by first performing a measurement and then choosing the corresponding $U_i$ gate using the measurement outcome. What we are applying here is the well known principle of deferred measurement \cite{book2000mikeandike}. This enables a further simplification of the circuit to that in Figure~\ref{fig:adaptive-circuit-2}.

Finally, we note that the circuit in Figure~\ref{fig:adaptive-circuit-2} consists of iterations of a single atomic unit, which we denote concisely in Figure~\ref{fig:adaptive-circuit-3}. This depiction of our state discrimination scheme is of particular import: to perform our state discrimination circuit, one only needs two quantum registers, the ability to perform a standard basis measurement, and classical control. This makes it directly amenable to experimental implementation.

\section{Error Analysis of the State Discrimination Circuit} \label{sec:error-analysis}

Here, we define the average probability of error in our state discrimination scheme. Next we demonstrate how it decays with $n$, the number of iterations of the state discrimination circuit in Figure~\ref{fig:dctc-circuit-simulation}. Finally, we show how the error probability exponentially converges to zero in the asymptotic limit.

\subsection{Calculating the Error Probability}

As earlier, we assume that we are given a set $\{|\psi_i\rangle\}_{i=0}^{N-1}$ of states, a set $\{|i\rangle \}_{i=0}^{N-1}$ of basis states and a set $\{U_i\}_{i=0}^{N-1}$ of unitaries such that $U_i|\psi_i\rangle = |i\rangle$ for $0 \leq i \leq N -1$. We assume that the states $\{|\psi_i\rangle\}_{i=0}^{N-1}$ are to be discriminated using the circuit shown in Figure $\ref{fig:dctc-circuit-simulation}$ with
\begin{equation}
 V_{SC} \coloneqq  \left(\sum_{i=0}^{N-1} |i\rangle\!\langle i| \otimes U_i\right)\circ \operatorname{SWAP}.   
\end{equation}
First, we provide an expression for the average probability of incorrectly discriminating the set of states after $n$ iterations of the circuit. 

\begin{definition}
For each $0\leq k \leq N-1$, let $\sigma_{n, k}$ be the state of the $C_n$ system when $\rho_S = |\psi_k\rangle\!\langle\psi_k|$. (That is, $\sigma_{n, k}$ is the state obtained after $n$ rounds of executing the circuit in Figure \ref{fig:dctc-circuit-simulation} with $\rho_S$ set equal to $|\psi_k\rangle\!\langle \psi_k|$.) We define the average probability of error $p_e^{(n)}$ and the average probability of success $p_s^{(n)}$ after $n$ iterations as follows:
\begin{align}
p^{(n)}_e &\coloneqq  \sum_{k=0}^{N-1} \sum_{j\neq k} p_k \operatorname{Tr}\big\{|j\rangle\!\langle j|\sigma_{n, k}\big\}, \\
p^{(n)}_s &\coloneqq  \sum_{k=0}^{N-1} p_k \operatorname{Tr}\big\{|k\rangle\!\langle k|
\sigma_{n, k}\big\}. 
\end{align}
        
\end{definition}

Our goal now is to quantitatively describe the behavior of $p^{(n)}_e$ with the number $n$ of iterations. We will find that the functioning of our circuit, particularly the error and success probabilities, possesses a Markov property. To obtain this property,  
for each $0\leq k\leq N-1$, we define the $N \times N$ stochastic matrix $P_k$ to be
\begin{equation} \label{eq:pk-definition}
P_k \coloneqq 
\begin{pmatrix}
|\langle 0 |U_0|\psi_k\rangle|^2 & \dots & |\langle 0 |U_{N-1}|\psi_k\rangle|^2 \\

\vdots & \ddots & \vdots \\

|\langle N-1 |U_0|\psi_k\rangle|^2 & \dots & |\langle N-1 |U_{N-1}|\psi_k\rangle|^2
\end{pmatrix}.
\end{equation}

Also, we define, for $0 \leq k \leq N - 1$, the $N \times 1$ column vector $u_k^{(n)}$ such that its $i$th element is 
\begin{equation}
\operatorname{Tr}\big\{|i-1\rangle\!\langle i-1| \sigma_{n, k}\}.    
\end{equation}
Define $u^{(0)}$ to be $u_k^{(0)}$ for any $0\leq k \leq N-1$. The vector $u^{(0)}$ is well defined since $\rho_{0,k} = \omega$ for all $0\leq k\leq N-1$, where $\omega$ is defined just before~\eqref{eq:controlled-unitary-defn}.

\begin{proposition}\label{prob-formula}
The average probability of error and success after $n$ iterations, denoted $p_e^{(n)}$ and $p_s^{(n)}$, are respectively given by
\begin{align}
p^{(n)}_e & =\sum_{k=0}^{N-1} \sum_{\substack{j=0, \\ j \neq k}}^{N-1} p_k e_{j+1}^T P_k^n u^{(0)} \label{eq:probability-error}, \\
p^{(n)}_s & = \sum_{k=0}^{N-1} p_k e_{k+1}^T P_k^n u^{(0)}, \label{eq:probability-success}
\end{align}
where $e_j$ is the standard basis column vector with a one in the $j$th row and all other elements set to zero.
\end{proposition}

\begin{proof}
The probability of measuring outcome $j$ on the $(n+1)$-th iteration of the circuit, assuming the unitary $V_{SC}$ always acts on the $S_n$ and $C_n$ registers, is
\begin{equation} \label{eq:probability-measurement-outcome}
\begin{split}
& \operatorname{Tr}\big\{|j\rangle\!\langle j|
\mathcal{N}_{V,\rho}(
\sigma_{n, k})\big\} \\
&  = 
\operatorname{Tr}\bigg\{|j\rangle\!\langle j|
\left(\operatorname{Tr}_{S} \big\{V_{SC} (|\psi_k\rangle\!\langle\psi_k|\otimes
\sigma_{n, k})V_{SC}^{\dagger}\big\}\right)\bigg\} \\
& = 
\operatorname{Tr}\bigg\{|j\rangle\!\langle j| \left(\sum_{l=0}^{N-1}
\operatorname{Tr} \big\{|l\rangle\!\langle l| \sigma_{n, k}\big\}
U_l|\psi_k\rangle\!\langle\psi_k|U_l^{\dagger}\right)\bigg\}\\
& =
\sum_{l=0}^{N-1}
\operatorname{Tr} \big\{|l\rangle\!\langle l| 
\sigma_{n, k}\big\}
|\langle j|
U_l|\psi_k\rangle|^2,
\end{split}
\end{equation}
where we used the definition of $\mathcal{N}_{V,\rho}$ given in~\eqref{eq:CTC-map-definition} in the first equality and the definition of $V_{SC}$ given in~\eqref{eq:brun-interaction-unitary} (including the SWAP operation) in the second equality.

It follows from \eqref{eq:probability-measurement-outcome} that $u_k^{(n+1)} = P_k u_k^{(n)}$. From this, we conclude that 
\begin{equation} \label{eq:uk-relation-to-u0}
 u_k^{(n)} = P_k^n u^{(0)}. 
\end{equation}
Since, $U_k|\psi_k\rangle = |k\rangle$, the ($k+1$)-th column of $P_k$ consists of zeroes except for a one in the $(k+1)$-th row.

Hence, we have 
\begin{align}
    p_e^{(n)} &= \sum_{k=0}^{N-1}\sum_{\substack{j = 0 : j\neq k}}^{N-1} p_k \operatorname{Tr} \{|j\rangle\!\langle j |\sigma_{n, k}\} \\
    &= \sum_{k=0}^{N-1} \sum_{\substack{j = 0 : j\neq k}}^{N-1} p_k e_{j+1}^T u_k^{(n)} \\
    &=  \sum_{k=0}^{N-1} \sum_{\substack{j = 0 : j\neq k}}^{N-1} p_k e_{j+1}^T P_k^n u^{(0)}.
\end{align}
In the above, the first equality arises due to the definition of $p_e^{(n)}$. The second equality is due to the definition of $u_k^{(n)}$, and the final equality is due to \eqref{eq:uk-relation-to-u0}. We also have that
\begin{align}
    p^{(n)}_s &= \sum_{k=0}^{N-1} p_k \operatorname{Tr}\big\{|k\rangle\!\langle k|
\sigma_{n, k}\big\} \\
&= \sum_{k=0}^{N-1} p_k e_{k+1}^T u_k^{(n)} \\
&= \sum_{k=0}^{N-1} p_k e_{k+1}^T P_k^n u^{(0)}.
\end{align}
This concludes the proof.
\end{proof}

The functioning of our quantum circuit has a Markov property, which we explain here. Recall from earlier that the state discrimination circuit in Figure~\ref{fig:dctc-circuit-simulation} can be implemented as an adaptive quantum state discrimination scheme (Figure~\ref{fig:adaptive-figure-group}). The adaptive scheme is as follows: one starts with $n$ copies of $\rho_S$. If one obtains outcome $k$ after measuring the $n$th copy of $\rho_S$, then one applies the unitary $U_k$ to the $(n+1)$-th copy of $\rho_S$, and repeats the procedure. Suppose that we have $\rho_S = |\psi_k\rangle\!\langle \psi_k|$ and that we obtain outcome $i$ after measuring the $n$th copy. The probability of obtaining the outcome $j$ after measuring the $(n+1)$-th copy is $|\langle j|U_i|\psi_k\rangle|^2$. That is, the probability of obtaining each successive measurement outcome given the previous measurement outcome depends only on the previous measurement outcome. Thus, we may view each successive measurement outcome as an element of a Markov chain with transition matrix $P_k$ as defined in $\eqref{eq:pk-definition}$.

We now state and prove a simple result that allows us to maximize the probability of success after $n$ iterations with respect to $\omega$, the initial state of the $C_0$ register.
\begin{proposition}\label{success-probability-omega}
For each $0\leq i \leq N-1$, let $p_{s,i}^{(n)}$ be the average probability of success after $n$ iterations given that $\omega = |i\rangle\!\langle i|$. Then the average probability of success (for arbitrary $\omega$) is 
\begin{equation} \label{eq:probability-success-average}
p_s^{(n)} = \sum_{i=0}^{N-1} \operatorname{Tr}\big\{|i\rangle\!\langle i|\omega\big\}p_{s,i}^{(n)}.
\end{equation}
\end{proposition}
\begin{proof}
It follows from Proposition~\ref{prob-formula} that 
\begin{align}\label{eq:recursive-prob-formula}
    p_{s}^{(n)} & = \sum_{k=0}^{N-1} p_k e_{k+1}^T P_k^n \bigg(\sum_{i=0}^{N-1} \operatorname{Tr} \{|i\rangle\!\langle i| \omega\} e_{i+1}\bigg) \\
    & = \sum_{i=0}^{N-1}\operatorname{Tr} \{|i\rangle\!\langle i|\omega \} \bigg(\sum_{k=0}^{N-1} p_k e_{k+1}^T P_k^ne_{i+1}\bigg) \\
    & = 
    \sum_{i=0}^{N-1} \operatorname{Tr} \{|i\rangle\!\langle i|\omega\} p_{s,i}^{(n)}.
\end{align}
The first equality arises due to Proposition~\ref{prob-formula} and the definition of $u^{(0)}$. The second equality arises due to algebraic manipulation, and the final equality follows from identifying that $p_{s,i}^{(n)} = \sum_{k=0}^{N-1} p_k e_{k+1}^T P_k^ne_{i+1}$.
\end{proof}

It follows from Proposition \ref{success-probability-omega} that the success probability after $n$ iterations is maximized if $\omega$ is simply a basis state of the form $|j\rangle\!\langle j|$. The particular optimal basis state is dependent on the value of $j$ that maximizes~$p_{s, j}^{(n)}$. 

\subsection{Asymptotic Analysis of Error Probability}

We now consider how the average probability of error decays with $n$. First we recall that each $P_k$ is the transition matrix of the Markov chain of successive measurement outcomes. We are interested in the equilibrium state, or steady state, of this Markov chain. It is a standard fact in Markov chain theory that the rate of convergence is determined by the second largest eigenvalue of the transition matrix. To utilize this fact, for $0 \leq k \leq N -1 $, we construct the $(N-1) \times (N-1)$ matrix $Q_k$ by deleting the $(k+1)$-th row and the $(k+1)$-th column of $P_k$. We will also construct, for $0 \leq k \leq N -1 $, column vectors $v_k^{(n)}$ by deleting the $(k+1)$-th entry in $u^{(n)}_k$. The vector $v^{(0)}$ is similarly constructed from $u^{(0)}$. This construction enables us to write the expression for $p_e^{(n)}$ in a more useful way. 

\begin{proposition}
After $n$ iterations of the state discrimination circuit, the average probability of error is given by
\begin{equation}
    p_e^{(n)} = \sum_{k=0}^{N-1} \sum_{j=1}^{N-1} p_k e_j^T Q_k^n v^{(0)}.
\end{equation}
\end{proposition}

\begin{proof}
We indicate here that this is a rewriting of Proposition~\ref{prob-formula}, in which we make use of the definitions introduced directly above. For completeness, we provide details below.

First, we have as a direct consequence of \eqref{eq:recursive-prob-formula} that
\begin{equation}
    v_k^{(n)} = Q_k^n v^{(0)}.
\end{equation}
Hence we have
\begin{align} \label{eq:Qk-probability-of-error}
    p_e^{(n)} &= \sum_{k=0}^{N-1}\sum_{\substack{j = 0 \\ j\neq k}}^{N-1} p_k \operatorname{Tr} \{|j\rangle\!\langle j |\sigma_{n, k}\} \\
    &= \sum_{k=0}^{N-1} \sum_{j=1}^{N-1} p_k e_{j}^T v_k^{(n)} \\
    &=  \sum_{k=0}^{N-1} \sum_{j=1}^{N-1} p_k e_{j}^T Q_k^n v^{(0)}.
\end{align}
In the above, the first equality is due to the definition of $p_e^{(n)}$. The second equality is due to the definition of $v_k^{(n)}$. The final equality is due to the fact that $v_k^{(n)} =Q_k^n v^{(0)}$, which is a direct consequence of \eqref{eq:uk-relation-to-u0}.
\end{proof}


Before we state our next result that quantifies the rate of decay of $p_e^{(n)}$, we establish some notation that we will use to prove it. First, for each $0\leq k \leq N-1$, let $s_k$ denote the number of distinct eigenvalues of $Q_k$ and let $\lambda_{1, k}, \dots, \lambda_{s_k, k}$ be the distinct eigenvalues of $Q_k$. That is, each $Q_k$ has eigenvalues $\lambda_{1, k},  \lambda_{2, k}, \dots, \lambda_{s_k, k} $, and let $m_{i, k}$ denote the algebraic multiplicity of eigenvalue $\lambda_{i,k}$. Also, let
\begin{equation}
\tau \coloneqq  \max_{i,k} | \lambda_{i,k} |.
\label{eq:def-tau}
\end{equation}
That is, $\tau$ is the largest absolute value of the eigenvalues of all the $Q_k$ matrices taken together for $0 \leq k \leq N - 1$. 

\begin{proposition}\label{error-exponent-proposition}
The asymptotic error exponent of the state discrimination scheme outlined above is not smaller than the negative logarithm of $\tau$, defined in \eqref{eq:def-tau}. That is, the following inequality holds:
\begin{equation} \label{eq:error-probability-scaling}
\xi \coloneqq  \lim_{n\to\infty} \frac{-\ln{p_e^{(n)}}}{n} \geq 
-\ln{\tau}.
\end{equation}
\end{proposition}

\begin{proof}
%
We begin by denoting the $m_{i,k}\times m_{i,k}$ Jordan block of $Q_k$ with eigenvalue $\lambda_{i,k}$ by $H_m(\lambda_{i,k})$, for each $1\leq i \leq s_k$. For each $0 \leq k \leq N-1$, define the $(N-1)\times (N-1)$ block-diagonal matrix 
\begin{gather}
J_k = 
\begin{pmatrix}
H_{m_{1, k}}(\lambda_{1,k}) & 0 & \dots & 0 \\
0 & H_{m_{2, k}}(\lambda_{2,k}) & \dots & 0 \\ 
\vdots& \vdots & \ddots & \vdots \\
0&0 & \dots & H_{m_{s_k, k}}(\lambda_{s_k, k}) 
\end{pmatrix}
\end{gather}
such that $Q_k = S_k J_k S_k^{-1}$ where $S_k$ is an invertible $(N-1)\times (N-1)$ matrix. Such a matrix exists since $J_k$ is the Jordan form of $Q_k$. Recall the expression for $p_e^{(n)}$ in ~\eqref{eq:Qk-probability-of-error}. We are interested in the matrix $Q_k^n$. Using the fact that $Q_k = S_k J_k S_k^{-1}$, we have that $Q_k^n = S_k J_k^n S_k^{-1}$. We then write, for $0\leq k \leq N-1$, $1\leq i \leq s_k$, the following \cite[p.~618]{Meyer}:
\begin{multline}\label{eq:Jordan-power}
 H_{m_{i, k}}^n (\lambda_{i,k}) = \\
 \begin{pmatrix}
\lambda_{i,k}^n & {\binom{n}{1}} \lambda_{i,k}^{n-1} & {\binom{n}{2}} \lambda_{i,k}^{n-2}  & \dots & {\binom{n}{m_{i,k} -1}} \lambda_{i,k}^{n-m_{i,k}+1} \\
 & \lambda_{i,k}^n & {\binom{n}{1}} \lambda_{i,k}^{n-1} 
 & \dots  & {\binom{n}{m_{i,k}-2}}\lambda_{i,k}^{n-m_{i,k}+2} \\
 & & \ \ \ddots & \ddots &  \vdots \\
 & \text{\huge0} & \ \ \ddots & \ \ \ddots & \vdots &  \\
 & & & \lambda_{i,k}^n & {\binom{n}{1}} \lambda_{i,k}^{n-1} \\
 & & & & \lambda_{i,k}^n 
 \end{pmatrix}
\end{multline}
where we have used the fact that, for $j>n$, ${\binom{n}{j}} = 0$.

Since for each integer $n$, we have
\begin{gather} 
J_k^n = 
\begin{pmatrix}
H_{m_{1, k}}^n(\lambda_{1,k}) & 0 & \dots & 0 \\
0 & H_{m_{2, k}}^n(\lambda_{2,k}) & \dots & 0 \\ 
\vdots& \vdots & \ddots & \vdots \\
0&0 & \dots & H_{m_{s_k, k}}^n(\lambda_{s_k, k}) 
\end{pmatrix},
\end{gather}
it follows from \eqref{eq:probability-error} and \eqref{eq:Jordan-power} that for sufficiently large $n$, we have 
\begin{equation} \label{eq:pe-expression-with-coefficients}
    p_e^{(n)} = \sum_{k=0}^{N-1}\sum_{\substack{i=1 \\ \lambda_{i,k \neq 1}}}^{s_k}\sum_{j=0}^{m_{i,k}-1} a_{i,j,k} {\binom{n}{j}}\lambda_{i, k}^{n-j}
\end{equation}
for some set $\{a_{i,j,k}\}_{i,j,k}$ of constants, where $a_{i,j,k}\in \mathbb{C}$ for all $i,j,k$. To understand the above, we recall that
\begin{align}
    p_e^{(n)} &= \sum_{k=0}^{N-1} \sum_{j=1}^{N-1} p_k e_{j}^T Q_k^n v^{(0)} \\
                &= \sum_{k=0}^{N-1} \sum_{j=1}^{N-1} p_k e_{j}^T S_k J_k^n S_k^{-1} v^{(0)}.
\end{align}
That is, $p_e^{(n)}$ is a linear combination of the elements of the matrices $J_k^n$ for $0\leq k \leq N-1$, and hence also a linear combination of powers of the eigenvalues $\lambda_{i,k}$. Further, by inspecting the elements of~\eqref{eq:Jordan-power}, the linear combination takes the form in~\eqref{eq:pe-expression-with-coefficients}.  
 
We then have that
\begin{align}
\label{eq:error-exponent-calculation}
    \xi &\coloneqq  -\lim_{n\to\infty} \frac{\ln{p_e^{(n)}}}{n}  \\ 
    &= -\lim_{n\to\infty}\frac{1}{n}\ln\!\left[\tau^n\left( \mathlarger{\sum}_{i,j,k} \frac{a_{i,j,k}}{\lambda_{i,k}^j} {\binom{n}{j}} \bigg(\frac{\lambda_{i,k}}{\tau}\bigg)^{n}\right)\right] \\ 
    &= -\ln\tau  - \lim_{n\to\infty} \frac{1}{n}\ln\!\left( \mathlarger{\sum}_{i,j,k} \frac{a_{i,j,k}}{\lambda_{i,k}^j} {\binom{n}{j}} \bigg(\frac{\lambda_{i,k}}{\tau}\bigg)^{n}\right) \label{eq:limit-second-term} \\ 
    &\geq -\ln{\tau}.
\end{align}
In the above, the first equality is due to \eqref{eq:pe-expression-with-coefficients}. The second equality is due to algebraic manipulation. To establish the final inequality, consider the following chain of reasoning:
\begin{align}
& \sum_{i,k}\sum_j\frac{a_{i,j,k}}{\lambda_{i,k}^{j}}\binom{n}{j}\left(
\frac{\lambda_{i,k}}{\tau}\right)  ^{n}\nonumber\\
& =\left\vert \sum_{i,k}\sum_j\frac{a_{i,j,k}}{\lambda_{i,k}^{j}}\binom{n}%
{j}\left(  \frac{\lambda_{i,k}}{\tau}\right)  ^{n}\right\vert \\
& \leq\sum_{i,k}\sum_j\frac{\left\vert a_{i,j,k}\right\vert }{\left\vert
\lambda_{i,k}\right\vert ^{j}}\binom{n}{j}\left(  \frac{\left\vert
\lambda_{i,k}\right\vert }{\tau}\right)  ^{n}\\
& =\sum_{\left(  i,k\right)  \in\mathcal{L}} \sum_{j} \frac{\left\vert
a_{i,j,k}\right\vert }{\left\vert \lambda_{i,k}\right\vert ^{j}}\binom{n}%
{j}\left(  \frac{\left\vert \lambda_{i,k}\right\vert }{\tau}\right)
^{n}\nonumber\\
& \qquad+\sum_{\left(  i,k\right)  \notin\mathcal{L}} \sum_{j} \frac{\left\vert
a_{i,j,k}\right\vert }{\left\vert \lambda_{i,k}\right\vert ^{j}}\binom{n}%
{j}\left(  \frac{\left\vert \lambda_{i,k}\right\vert }{\tau}\right)  ^{n}\\
& =\sum_{\left(  i,k\right)  \in\mathcal{L}} \sum_{j} \frac{\left\vert
a_{i,j,k}\right\vert }{\left\vert \lambda_{i,k}\right\vert ^{j}}\binom{n}%
{j}+e^{-\Omega(n)}\\
& =O(\text{poly}(n))+e^{-\Omega(n)},
\end{align}
where $\mathcal{L}$ denotes the set of pairs $\left(  i,k\right)  $\ for which
$\left\vert \lambda_{i,k}\right\vert =\tau$ (thus, if $\left(  i,k\right)
\notin\mathcal{L}$, then $\left\vert \lambda_{i,k}\right\vert <\tau$). In the
above, we have employed the triangle inequality and the fact that $\binom
{n}{j}\leq\frac{n^{j}}{j!}$. By applying the negative logarithm and using its anti-monotonicity, normalizing, and taking the limit $n\to\infty$, it follows that
\begin{multline}
    -\lim_{n \to \infty}\frac{1}{n}\ln\left(\sum_{i,k}\sum_{j} \frac{a_{i,j,k}}{\lambda_{i,k}^{j}}\binom{n}{j}\left(
\frac{\lambda_{i,k}}{\tau}\right)  ^{n}\right)
\geq \\
-\lim_{n \to \infty} \frac{1}{n} \ln\left(O(\text{poly}(n))+e^{-\Omega(n)}\right) = 0.
\end{multline}
This establishes the desired inequality
\begin{equation}
    \xi \geq -\ln \tau
\end{equation}
and concludes the proof.
\end{proof}

In Appendix~\ref{app:xi-bounds}, we state and prove simple lower bounds on $\xi$ in terms of the unitaries $\{ U_i \}_i$ and the states $\{ | \psi_i \rangle \}_i$.

\section{Examples} \label{sec:examples}

We now discuss how to optimize the performance of our state discrimination circuit in specific cases. To optimize the performance of our state discrimination circuit, it is necessary to find a set of unitaries that minimizes the probability of error. We may find expressions for these unitaries in simple cases, but this becomes difficult in the general case. An alternative route is to maximize the error exponent $\xi\coloneqq  -\lim_{n\to \infty} \frac{\ln{p_e^{(n)}}}{n}$. In the following, we discuss explicit state discrimination schemes for sets of qubit states.

\subsection{Two Qubit States}

Our first example is the simplest possible, where we consider that we are to discriminate between two pure qubit states $|\psi_0\rangle$ and $|\psi_1\rangle$.

To perform the state discrimination, we require unitaries $U_0$ and $U_1$ such that $U_0|\psi_0\rangle = |0\rangle$ and $U_1|\psi_1\rangle = |1\rangle$. We may write the two unitaries $U_0$ and $U_1$ in the form 
\begin{equation} \label{eq:qubit-unitaries}
\begin{split}
U_0 & = e^{i\phi_{0}} |0\rangle\!\langle\psi_0| 
+ e^{i\phi_{1}} |1\rangle\!\langle\psi_0^{\perp}|  \\
U_1 & = e^{i\phi_{2}} |1\rangle\!\langle\psi_1|
+ e^{i\phi_{3}} |0\rangle\!\langle\psi_1^{\perp}|
\end{split}
\end{equation}
where $\ket{\psi_0^{\perp}}$ and $\ket{\psi_1^{\perp}}$ are pure states orthogonal to $\ket{\psi_0}$ and $\ket{\psi_1}$, respectively.
We then have
\begin{equation}
\begin{split}
Q_0 & = |\langle 1|U_1|\psi_0\rangle|^2 = |\langle \psi_0|\psi_1\rangle|^2 \quad \text{and}\\
Q_1 & = |\langle 0|U_0|\psi_1\rangle|^2 = |\langle \psi_0|\psi_1\rangle|^2.
\end{split}
\end{equation}

From Proposition~\ref{error-exponent-proposition} and the Chernoff bound \cite{NS11,Li-Chernoff-Bound}, we have $\xi = -\ln{|\langle \psi_0|\psi_1\rangle |^2}$. We see that the average error probability $p_e^{(n)}$ is independent of the choice of the unitaries, as expected from Proposition~\ref{prob-formula}. Further, the error probability $p_e^{(n)}$ scales according to the Chernoff bound in~\eqref{eq:chernoff-bound-def}. 

\subsection{Arbitrary Set of Qubit States} \label{subsec:arbitrary-qubit-states}

We now show how to construct a set of unitaries for discriminating an arbitrary set of more than two pure qubit states, and we find that our construction ensures that the probability of error $p_e^{(n)}$ scales according to the multiple Chernoff bound, generalizing what we showed above for two qubit states. This means that the probability of error decays at the optimal rate, so that that our state discrimination scheme will perform better than or as well as any other scheme designed to discriminate qubit states in the asymptotic case.  

Let $\{|\psi_i\rangle\}_{i=0}^{N-1}$ be a set of $N>2$ qubit states, each of which is in a two-dimensional Hilbert space $\mathcal{H}$. Let $\{|i\rangle\}_{i=0}^{N-1}$ be a basis for an $N$-dimensional Hilbert space $\mathcal{H}'$.  
We will use the following isometries to perform our state discrimination protocol:
\begin{equation}\label{eq:isometry-definition}
    V_i = |i\rangle\!\langle \psi_i| + |i\oplus 1\rangle\!\langle \psi_i^{\perp}| ,
\end{equation}
where $| \psi_i^{\perp} \rangle $ is a pure state orthogonal to $|  \psi_i \rangle$. Note that each $V_i$ is an isometry mapping $\mathcal{H}$ to $\mathcal{H}'$ and satisfies $V_i | \psi_i\rangle = \ket{i}$. Given each isometry $V_i$, let  $U_i$  be its unitary extension, satisfying
\begin{equation}
    V_i \ket{\psi} = U_i \ket{\tilde{\psi}} 
\end{equation}
for every $\ket{\psi}\in \mathcal{H}$ and where $\ket{\tilde{\psi}} \in \mathcal{H}'$ denotes an embedding of $\ket{\psi}$ in $\mathcal{H}'$. Note that if $N=2^\ell$ for some integer~$\ell$, then we can set $\ket{\tilde{\psi}} = \ket{\psi} \ket{0}^{\otimes (\ell-1)}$.

We then get that for $0\leq k \leq N-1$, the matrix $P_k$ takes on the form 
\begin{multline}
P_k = \\
\begin{pmatrix} 
|\langle \psi_0|\psi_k\rangle|^2 & 0 & \hdots &  1-|\langle \psi_{N-1} |\psi_k\rangle|^2 \\
1-\langle \psi_0|\psi_k\rangle|^2 & |\langle \psi_1|\psi_k\rangle|^2 &    & 0 \\
0 & 1 - |\langle \psi_1|\psi_k\rangle|^2 & \ddots & \vdots \\
\vdots & \vdots & \ddots  & 0 \\
0 & 0 & \hdots &  |\langle \psi_{N-1}|\psi_k\rangle|^2
\end{pmatrix}.
\end{multline}

To construct each matrix $Q_k$, we delete the $(k+1)$-th rows and columns from $P_k$. Recall that the $(k+1)$-th column of each $P_k$ contains all zeroes except for a one in the $(k+1)$-th row. This means that $Q_0$ and $Q_{N-1}$ will be lower-triangular matrices. The other $Q_k$ matrices will consist of a diagonal, a sub-diagonal containing a zero element, and a non-necessarily-zero element $(Q_k)_{1,(N-1)}$, with all other elements set to zero. For such matrices, the eigenvalues are given by their diagonal entries. This can be seen by writing out the characteristic polynomial of the matrix, and taking care to expand out the determinant along the row or column containing the zero element of the sub-diagonal. Therefore, the largest eigenvalue of each matrix $Q_k$ is equal to $\max_{(i,j)} |\langle \psi_i | \psi_j\rangle |^2$. By Proposition~\ref{error-exponent-proposition} and the multiple Chernoff bound \cite{NS11,Li-Chernoff-Bound}, it follows that \linebreak $\xi = -\ln{\max_{i\neq j}} \{|\langle \psi_i|\psi_j\rangle|^2\}$.

Therefore, by appending sufficiently many ancillary qubits, our state discrimination circuit can be used to discriminate an arbitrary set of qubit states with the optimal scaling of the probability of error given by the multiple Chernoff bound in~\eqref{eq:chernoff-bound-def}. This result only concerns the scaling of the probability of error in the asymptotic case. Investigating how close to optimal the probability of error is when only a finite number of copies of the input state are available may prove to be a fruitful direction for future work. 

Note that one may not extend in a straightforward way the above procedure for constructing a set of unitaries $\{U_i\}$ that produce optimal scaling of the probability of error to states in a Hilbert space of dimension greater than two. Observe that the above procedure hinges on the matrices $Q_k$ having eigenvalues given by their diagonal. By the definition of the $Q_k$ matrix, specifying a set of isometries in a way similar to that of \eqref{eq:isometry-definition} for a higher dimensional system would require that other elements besides $(Q_k)_{1, N-1}$ and the elements on the diagonal and subdiagonal of $Q_k$ be nonzero. Such matrices do not in general have eigenvalues equal to their diagonal elements.  

When simulating a D-CTC with unitary interaction given by $V_{SC} = (\sum_{k=0}^{N-1} |k\rangle\!\langle k|\otimes U_k)\circ \operatorname{SWAP}$, it follows from our choice of the set $\{U_k\}_k$ that $|j\rangle\!\langle j|$ is the unique solution for $\sigma_C$ in~\eqref{eq:self-consistency-condition} whenever $\rho_S = |\psi_j\rangle\!\langle \psi_j|$.
In Appendix~\ref{app:uniqueness-qubit-discrimination}, we provide a proof inspired by the argument given in \cite{Heisenberg-Holevo-Violation-BHW} that this is indeed true for the example considered above, consisting of discriminating an arbitrary set of qubit states.

\subsection{BB84 States}


Here, we study how our state discrimination circuit can be used to discriminate a specific set of non-orthogonal states, i.e., the BB84 states $|0\rangle, |1\rangle, |+\rangle$, and $|-\rangle$. In general, identifying or performing the set of unitaries $\{U_i\}$ described above may be difficult. The authors of \cite{Heisenberg-Holevo-Violation-BHW} identified a set of unitaries that can discriminate the BB84 states, which we restate here. Building off of these authors' work, we show that the probability of error using this construction scales according to the multiple Chernoff bound in~\eqref{eq:chernoff-bound-def}. This is of significance because these operators provide an example of unitaries that both produce a probability of error that saturates the multiple Chernoff bound and are constructed from well-studied, standard quantum logic gates. 

As we described in the example earlier, we encode these four states into a four-dimensional Hilbert space.
Let $|\psi_0\rangle \equiv |00\rangle$, $|\psi_1\rangle \equiv |10\rangle$, $|\psi_2\rangle \equiv \ket{+0}$, and $|\psi_3\rangle \equiv \ket{-0}$. That is, we obtain the set of states $\{|\psi_i\rangle\}_{i=0}^{3}$ by appending an ancillary qubit in the $|0\rangle$ state. Let $|0\rangle \equiv |00\rangle, |1\rangle \equiv |01\rangle, |2\rangle \equiv |10\rangle, |3\rangle \equiv |11\rangle$.
Now let 
\begin{equation}
\begin{split}
    U_0 &= \operatorname{SWAP} , \\ 
    U_1 &= X \otimes X , \\ 
    U_2 &= (X\otimes I)\circ (H\otimes I), \\
    U_3 &= (X\otimes H)\circ \operatorname{SWAP}.
\end{split}
\end{equation}
Then $U_i|\psi_i\rangle = |i\rangle$ for $i \in \{0,1,2,3\}$.
It may be seen that the probability of error $p_e^{(n)}$ scales according to the Chernoff bound by constructing the $Q_k$ matrices for $0\leq k \leq 3$ and checking that their largest eigenvalues are each $\max_{i\neq j} |\langle \psi_i|\psi_j\rangle|^2 = 1/2$. 

Since an arbitrary set of four geometrically uniform qubit states is simply a rotation of the BB84 states on the Bloch sphere, one may saturate the Chernoff bound when using our state discrimination circuit to  distinguish any set of four geometrically uniform qubit states using only compositions of the standard qubit gates $X$, $H$, $\operatorname{SWAP}$, and rotations.

\section{Concluding Remarks}

In this paper, we have proposed a method for discriminating multiple non-orthogonal states, which is inspired by a construction considered in the context of closed timelike curves \cite{Heisenberg-Holevo-Violation-BHW}. Our state discrimination method can be equivalently recast as a local, iterative circuit whose simplicity lends itself to experimental implementation. Furthermore, we studied the average probability of error for our scheme and showed that in the general case of discriminating an arbitrary set of pure qubit states, it achieves the multiple Chernoff bound. 

We would like to point out three aspects of our work that require further investigation. It has been shown that a two-state local adaptive state discrimination scheme may be optimal for any number of copies~\cite{Two-State-Disc-Acin-et-al}. It remains open whether there exists a way to configure some aspect of our cirucit differently so that it is possible for our scheme to be optimal for any number of copies.
Also, while sets of unitaries do exist that optimize the performance of our circuit in the asymptotic limit for an arbitrary set of qubit states, it is unknown whether there exists a set of optimal unitaries for any set of qubit states. Furthermore, it is worth investigating this aspect of our work to see whether there exist a set of \textit{product} unitaries to optimize, or even make the performance of our circuit sub-optimally efficient, in the asymptotic limit. This would be beneficial because product operators are more convenient to implement experimentally.


Finally, any attempt at practical state discrimination will be subject to noisy conditions. Noise may have the ability to enhance or worsen the performance of a state discrimination scheme. A recent result due to~\cite{Noisy-Qubits-Flatt-et-al} has shown that the optimal measurement in the discrimination of two pure qubit states is no longer optimal when these qubit states are subject to perturbations. Studying the behavior of our state discrimination circuit in the presence of noise remains a topic for future work. 

\begin{acknowledgments}
We thank Osa Adun, Todd Brun, and Eneet Kaur for discussions. This research was supported by the NSF through the LSU Physics and Astronomy REU program (NSF Grant No.~1852356). VK acknowledges support from the LSU Economic Development Assistantship. VK and MMW acknowledge support from the US National Science Foundation via grant number 1907615. MMW acknowledges support from AFOSR (FA9550-19-1-0369). 
\end{acknowledgments}
 
\bibliographystyle{unsrt}

\pagebreak

\appendix
\onecolumngrid

\section{Upper and Lower Bounds on the Multiple Chernoff Exponent $\xi$} \label{app:xi-bounds}

We now state a result that allows us to bound the scaling of the average probability of error in terms of the set $\{U_i\}_i$ of unitaries  and the set $\{|\psi_i\rangle\}_i$ of states.  
\begin{proposition}\label{xi-estimate}
The following inequalities hold:
\begin{align}
    \xi & \geq -\ln\!\left(1-\min_{\substack{j,k:j\neq k}} |\langle k|U_j|\psi_k\rangle|^2\right)  \\
    \xi & \geq -\ln\!\left(\max_{\substack{j,k:j\neq k}} \sum_{i\neq k} |\langle j|U_i|\psi_k\rangle|^2\right),
\end{align}
where $\xi$ is defined in \eqref{eq:error-exponent-calculation}.
\end{proposition}
\begin{proof}
Let $0\leq k \leq N-1$. The maximum column sum of $Q_k$ is \begin{equation}
    \max_{j\neq k} \sum_{i\neq k} |\langle i|U_j|\psi_k\rangle|^2 = 1-\min_{j\neq k} |\langle k|U_j|\psi_k\rangle|^2 .
\end{equation} 
It follows from the Gerschgorin Circle Theorem~\cite{Bhatia} that 
the largest eigenvalue of $Q_k$ is bounded 
from above by the maximum column sum of $Q_k$. We recall that $\tau \coloneqq  \max_{i,k} |\lambda_{i,k}|$, i.e. $\tau$ is the largest of the absolute values of the eigenvalues of all of the $Q_k$ matrices. It follows that 
\begin{equation}
    \tau \leq 1-\min_{\substack{j,k:j\neq k}} |\langle k|U_j|\psi_k\rangle|^2.
\end{equation}
Then using Proposition~\ref{error-exponent-proposition}, we have
\begin{equation}
    \xi \geq -\ln \tau \geq 
    -\ln\!\left(1-\max_{\substack{j,k:j\neq k}} |\langle k|U_j|\psi_k\rangle|^2\right).
\end{equation}
It also follows from the Gerschgorin Circle Theorem that for $0\leq k \leq N-1$, the largest eigenvalue of $Q_k$ is bounded from 
above by the 
maximum row sum of $Q_k$. Hence, we have 
\begin{equation}
    \tau \leq 
    \max_{\substack{j,k\\j\neq k}} \sum_{i\neq k} |\langle j|U_i|\psi_k\rangle|^2,
\end{equation}
so that 
\begin{equation}
    \xi \geq -\ln \tau \geq 
    -\ln\!\left(\max_{\substack{j,k:j\neq k}} \sum_{i\neq k} |\langle j|U_i|\psi_k\rangle|^2\right) %
\end{equation}
This concludes the proof.
\end{proof}

\section{Uniqueness of the Fixed Point of the Channel~$\mathcal{N}_{V,\rho}$} \label{app:uniqueness-qubit-discrimination}

Here, we show the uniqueness of the fixed point of the channel $\mathcal{N}_{V,\rho}$ when discriminating an arbitrary set of qubit states, which we studied in Section~\ref{subsec:arbitrary-qubit-states}.

\begin{proposition}
Let $N\geq 3$. For each $0\leq i \leq N-1$, let $U_i$ be a unitary extension of the isometry $V_i$ defined in $\eqref{eq:isometry-definition}$. Let $V_{SC} = (\sum_k |k\rangle\!\langle k|\otimes U_k)\circ \operatorname{SWAP}$. Then for $0 \leq a \leq N- 1$, $|a\rangle\!\langle a |$ is the unique solution for $\sigma_C$ in $\eqref{eq:self-consistency-condition}$ whenever $\rho_S = |\psi_a \rangle\!\langle \psi_a|$. 
\end{proposition}

\begin{proof} 
We argued that $|a\rangle\!\langle a|$ is a solution for $\sigma_C$ whenever $\rho_S = |\psi_a \rangle\!\langle \psi_a|$ in Section \ref{subsec:prelim-CTC-description}. It remains to show uniqueness. Suppose $\rho_S = |\psi_a\rangle\!\langle\psi_a|$ and $\sigma$ is a solution for $\sigma_C$ in~\eqref{eq:self-consistency-condition}. Then we have
\begin{align}
    \sigma &= \operatorname{Tr} _{S }\bigg\{V_{SC }(|\psi_a\rangle\!\langle \psi_a| \otimes \sigma) V_{SC }^\dagger \bigg\}\\ \nonumber
    &= \operatorname{Tr} _{S }\bigg\{\bigg(\sum_k |k\rangle\!\langle k| \otimes U_k\bigg)(\sigma \otimes |\psi_a\rangle\!\langle \psi_a| )\bigg(\sum_l |l\rangle\!\langle l| \otimes U_l^\dagger \bigg) \bigg\} \\  \nonumber
    &= \sum_{k, l} \operatorname{Tr} _{S } \bigg\{|k\rangle\!\langle k|\sigma|l\rangle\!\langle l| \otimes U_k|\psi_a\rangle\!\langle \psi_a| U_l^\dagger \bigg\} \\ \nonumber
    &= \sum_k \langle k |\sigma| k \rangle U_k|\psi_a\rangle\!\langle\psi_a| U_k^\dagger.
\end{align}
In the above, the first line is due to the self-consistency condition~\eqref{eq:self-consistency-condition}. The following equalities come from explicitly writing out $V_{SC}$ and algebraic manipulation.

Hence, the matrix elements of $\sigma$ are given by
\begin{equation}\label{eq:rho-matrix-el}
    \langle m |\sigma| n \rangle = 
    \sum_{k} \langle k|\sigma| k \rangle\!\langle m|U_k|\psi_a\rangle\!\langle \psi_a|U_k^\dagger|n\rangle.
\end{equation}

We now show that all diagonal elements of $\sigma$ other than $\langle a|\sigma|a\rangle$ are zero. We will proceed by induction to show that $\langle a\ominus s| \sigma |a\ominus s\rangle= 0$ for all $1\leq s \leq N-1$. (Here, $\ominus$ denotes subtraction modulo $N$.) We first show that $\langle a\ominus 1|\sigma|a\ominus 1\rangle = 0$. It follows from~\eqref{eq:rho-matrix-el} that 
\begin{equation}
    \langle a|\sigma| a \rangle = \langle a|\sigma| a \rangle + \sum_{k\neq a} \langle k|\sigma| k \rangle |\langle a|U_k|\psi_a\rangle |^2,
\end{equation}
which implies that $\sum_{k\neq a} \langle k|\sigma| k \rangle |\langle a|U_k|\psi_a\rangle|^2 = 0$. It follows that
\begin{equation}
\begin{split}
    \langle a\ominus 1|\sigma|a\ominus 1\rangle(1-|\langle \psi_{a\ominus 1}|\psi_a\rangle |^2) 
    &=
    \langle a\ominus1|\sigma|a\ominus1\rangle|\langle \psi_{a\ominus1}^\perp | \psi_a \rangle |^2 \\ &=
    \langle a\ominus1|\sigma|a\ominus1\rangle |\langle a | U_{a\ominus 1}|\psi_a\rangle |^2 \\&=
    \sum_{k\neq a} \langle k|\sigma|k \rangle |\langle a|U_k|\psi_a\rangle|^2 
    \\&= 0.
\end{split}
\end{equation}
Since $1-|\langle \psi_{a\ominus 1}|\psi_a\rangle |^2\neq 0$, we then have that $\langle a\ominus1|\sigma|a\ominus1 \rangle = 0$. Now suppose $1\leq s < N-1$ is such that $\langle a\ominus s|\sigma|a\ominus s\rangle = 0$. Using $\eqref{eq:rho-matrix-el}$ and the definition of the unitaries $\{U_i\}$, we have
\begin{equation}
\begin{split}
    \langle a\ominus s|\sigma|a\ominus s\rangle &= \sum_k \langle k|\sigma|k \rangle |\langle a \ominus s|U_k|\psi_a \rangle |^2 \\
    & = \langle a\ominus s|\sigma|a\ominus s\rangle|\langle a\ominus s|U_{a\ominus s}|\psi_a\rangle|^2 
     +
    \langle a\ominus(s+1)|\sigma|a\ominus(s+1)\rangle |\langle a\ominus s|U_{a\ominus(s+1)}|\psi_a\rangle|^2.
\end{split}
\end{equation}
Since $\langle a\ominus s|\sigma|a\ominus s\rangle = 0$, it follows that  
\begin{equation}\label{eq:ind-pf-final-step}
\begin{split}
    0 &= \langle a\ominus(s+1)|\sigma|a\ominus(s+1)\rangle |\langle a\ominus s|U_{a\ominus(s+1)}|\psi_a\rangle|^2 \\
    &= \langle a\ominus(s+1)|\sigma|a\ominus(s+1)\rangle
    |\langle \psi_{a\ominus(s+1)}^\perp|\psi_a\rangle|^2 \\
    &= \langle a\ominus(s+1)|\sigma|a\ominus(s+1)\rangle
    (1-|\langle \psi_{a\ominus(s+1)}|\psi_a \rangle|^2).
\end{split}
\end{equation}
Since $1\leq s< N-1$, we must have $a\ominus(s+1) \neq a$, so that $1- |\langle \psi_{a\ominus(s+1)}|\psi_a\rangle |^2\neq 0$. It then follows from \eqref{eq:ind-pf-final-step} that $\langle a\ominus(s+1)|\sigma|a\ominus(s+1)\rangle = 0$. By induction, we have $\langle a\ominus s|\sigma|a\ominus s\rangle = 0$ for $1\leq s \leq N-1$. Hence, all diagonal elements of $\sigma$ other than $\langle a|\sigma| a\rangle$ are zero. 

Consequently, we must have $\langle k|\sigma|k\rangle = \delta_{ka}$ since $\operatorname{Tr}  [\sigma] = 1$. Any density operator of this form has off-diagonal elements that are all zero. Therefore, we have $\sigma = |a\rangle\!\langle a|$. Hence, $|a\rangle\!\langle a|$ is the unique solution for $\sigma_C$ in $\eqref{eq:self-consistency-condition}$.
\end{proof}

\end{document}